\documentclass[twocolumn,superscriptaddress,aps,pra,10pt]{revtex4-1}
\usepackage{amssymb,amsmath,amsthm}
\usepackage{graphicx}
\usepackage{color}
\usepackage{hyperref}
\usepackage{soul}
\usepackage[dvipsnames]{xcolor}
\usepackage{ulem}
\newcommand{\ket}[1]{| #1 \rangle}

\def\identity{\leavevmode\hbox{\small1\kern-3.8pt\normalsize1}}

\newtheorem{lemma}{Lemma}

\begin{document}

\title{Cooperation and dependencies in multipartite systems}

\author{Waldemar K{\l}obus}
\affiliation{Institute of Theoretical Physics and Astrophysics, Faculty of Mathematics, Physics and Informatics, University of Gda\'nsk, 80-308 Gda\'nsk, Poland}

\author{Marek Miller}
\affiliation{School of Physical and Mathematical Sciences, Nanyang Technological University, 637371 Singapore}

\author{Mahasweta Pandit}
\affiliation{Institute of Theoretical Physics and Astrophysics, Faculty of Mathematics, Physics and Informatics, University of Gda\'nsk, 80-308 Gda\'nsk, Poland}

\author{Ray Ganardi}
\affiliation{Institute of Theoretical Physics and Astrophysics, Faculty of Mathematics, Physics and Informatics, University of Gda\'nsk, 80-308 Gda\'nsk, Poland}
\affiliation{International Centre for Theory of Quantum Technologies, University of Gda\'nsk, 80-308 Gda\'nsk, Poland}

\author{Lukas Knips}
\affiliation{Max-Planck-Institut f\"{u}r Quantenoptik, Hans-Kopfermann-Stra{\ss}e 1, 85748 Garching, Germany}
\affiliation{Department f\"{u}r Physik, Ludwig-Maximilians-Universit\"{a}t, Schellingstra{\ss}e 4, 80799 M\"{u}nchen, Germany}
\affiliation{Munich Center for Quantum Science and Technology (MCQST), Schellingstra{\ss}e 4, 80799 M\"{u}nchen, Germany}

\author{Jan Dziewior}
\affiliation{Max-Planck-Institut f\"{u}r Quantenoptik, Hans-Kopfermann-Stra{\ss}e 1, 85748 Garching, Germany}
\affiliation{Department f\"{u}r Physik, Ludwig-Maximilians-Universit\"{a}t, Schellingstra{\ss}e 4, 80799 M\"{u}nchen, Germany}
\affiliation{Munich Center for Quantum Science and Technology (MCQST), Schellingstra{\ss}e 4, 80799 M\"{u}nchen, Germany}

\author{Jasmin Meinecke}
\affiliation{Max-Planck-Institut f\"{u}r Quantenoptik, Hans-Kopfermann-Stra{\ss}e 1, 85748 Garching, Germany}
\affiliation{Department f\"{u}r Physik, Ludwig-Maximilians-Universit\"{a}t, Schellingstra{\ss}e 4, 80799 M\"{u}nchen, Germany}
\affiliation{Munich Center for Quantum Science and Technology (MCQST), Schellingstra{\ss}e 4, 80799 M\"{u}nchen, Germany}

\author{Harald Weinfurter}
\affiliation{Max-Planck-Institut f\"{u}r Quantenoptik, Hans-Kopfermann-Stra{\ss}e 1, 85748 Garching, Germany}
\affiliation{Department f\"{u}r Physik, Ludwig-Maximilians-Universit\"{a}t, Schellingstra{\ss}e 4, 80799 M\"{u}nchen, Germany}
\affiliation{Munich Center for Quantum Science and Technology (MCQST), Schellingstra{\ss}e 4, 80799 M\"{u}nchen, Germany}

\author{Wies{\l}aw Laskowski}
\affiliation{Institute of Theoretical Physics and Astrophysics, Faculty of Mathematics, Physics and Informatics, University of Gda\'nsk, 80-308 Gda\'nsk, Poland}
\affiliation{International Centre for Theory of Quantum Technologies, University of Gda\'nsk, 80-308 Gda\'nsk, Poland}

\author{Tomasz Paterek}
\affiliation{Institute of Theoretical Physics and Astrophysics, Faculty of Mathematics, Physics and Informatics, University of Gda\'nsk, 80-308 Gda\'nsk, Poland}
\affiliation{School of Physical and Mathematical Sciences, Nanyang Technological University, 637371 Singapore}
\affiliation{MajuLab, International Joint Research Unit UMI 3654, CNRS, Universit\'e C\^ote d'Azur, Sorbonne Universit\'e, National University of Singapore, Nanyang Technological University, Singapore}

\begin{abstract}
We propose an information-theoretic quantifier for the advantage gained from cooperation that captures the degree of dependency between subsystems of a global system.
The quantifier is distinct from measures of multipartite correlations despite sharing many properties with them.
It is directly computable for classical as well as quantum systems and reduces to comparing the respective conditional mutual information between any two subsystems.
{\color{black} Exemplarily we show the benefits of using the new quantifier for symmetric quantum secret sharing.}
We also prove an inequality characterizing the lack of monotonicity of conditional mutual information under local operations and provide intuitive understanding for it.
{\color{black} This underlines the distinction between the multipartite dependence measure introduced here and multipartite correlations.}
\end{abstract}

\maketitle

\section{Introduction}

Identifying and quantifying dependencies in multipartite systems enables their analysis and provides a better understanding of complex phenomena.
The problem has been addressed by several communities, considering both classical and quantum systems.
For example, in neuroscience and genetics measures of multipartite synergy were put forward~\cite{syn1,syn2,syn3,syn4,synergia1,synergia2},
in quantitative sociology quantifiers of coordination were introduced~\cite{socio},
{\color{black} redundancy was quantified in complex systems~\cite{WB2010},}
and in physics and information processing quantities aimed at characterizing genuine multiparty correlations were studied in depth~\cite{Zhou,Zhou2008,Kaszlikowski2008,Grudka,Giorgi,Girolami2017}.
The former quantifiers are motivated mathematically, keeping the combinatorial aspects of complex systems in mind, e.g., the synergy is the difference in the information all subsystems have about an extra system as compared to the total information contained in any subset of the systems.
Many of the latter quantifiers involve difficult optimizations and are therefore hard to compute,
{\color{black} e.g. in order to compute the extractable work used to define genuinely multipartite correlations in \cite{Grudka} one has to optimise over all protocols with local unitary operations, local dephasings and classical communication, which does not admit any computer-friendly parameterisation.}
Here, we introduce an operationally defined, simple and computable quantifier of multipartite dependency in terms of information gain from cooperation when some parties meet and try to deduce the variables of some of the remaining parties.
We show how it differs from multipartite correlations {\color{black} often discussed in this context}, prove its essential properties and discuss {\color{black} examples and applications}.

It turns out that, in order to compute the quantity introduced here, it is sufficient to consider the respective conditional mutual information between only two subsystems.
{\color{black} The conditional mutual information is a well established quantity in quantum information theory and broader physics.
It captures the communication cost of quantum state redistribution~\cite{Devetak2003,Devetak2008,Devetak2009,Brandao2015}, topological entanglement entropy~\cite{Kitaev2006,Levin2006,Kim2012,Zeng2018} and squashed entanglement~\cite{Christandl2004}.
It is related to quantum discord~\cite{Piani2012,Byrnes2020}, intrinsic steerability~\cite{Kaur2017}, conditional erasure cost of a tripartite state~\cite{Berta2018,Berta20182} and the conditional quantum one-time pad~\cite{Sharma2020}.
Additionally, conditional quantum mutual information appears in thermodynamics~\cite{Mahajan2016}, in relation to high-energy physics~\cite{Czech2015,Ding2016,Pastawski2017} and Markov chains~\cite{Petz1986,Ruskai2002,Petz2003,Ibinson2007,Fawzi2015,Sutter2016}.
Therefore, the dependence introduced here is a relevant quantity in all these problems if we symmetrize them and look for the worst case scenario.
Concrete examples are given below.
Furthermore, all these different studies will benefit from the inequality we prove here, which} characterizes the lack of monotonicity of quantum conditional mutual information under general local operations.


\section{Multipartite dependence}

Let us begin by briefly recalling fundamental relationships, e.g., that two classical variables $X_1$ and $X_2$ are statistically independent if their probabilities satisfy $P(X_1 | X_2) = P(X_1)$. 
Alternatively, the statistical independence can be stated in terms of entropies with the help of both the Shannon entropy {\color{black} \cite{NielsenChuang}}
 $H(X) = -\sum^d_{i=1} P(x_i)\log_d P(x_i)$, where $d$ is the number of outcomes, and the conditional entropy $H(X|Y) = -\sum_{i,j}P(x_i,y_j) \log_d \frac{P(x_i,y_j)}{P(y_j)}$. {\color{black}Note that throughout this paper logarithms are base $d$.}  As a measure of dependence of two variables $X_1$ and $X_2$ one introduces the corresponding entropic difference $H(X_1) - H(X_1 | X_2)$, the so-called mutual information $I(X_1 : X_2)$~\cite{CoverThomas}.
 Similarly, the quantum mutual information captures the dependence between quantum subsystems~\cite{Modi2010}, characterized by the state $\rho$, when we use the von Neumann entropy $S(\rho) = - \textrm{Tr}{(\rho \log_d{\rho})}$ in place of the Shannon entropy.
However, already in the case of three variables there are two levels of independence.
The variable $X_1$ can be independent of all other variables, i.e., $P(X_1 | X_2 X_3) = P(X_1)$, or it can be conditionally independent of one of them, e.g., $P(X_1 | X_2 X_3) = P(X_1 | X_2)$.
The former dependence is again captured by the mutual information $I(X_1 : X_2 X_3)$, while the so-called conditional mutual information $I(X_1 : X_3 | X_2)=H(X_1|X_2) - H(X_1 | X_2 X_3)$ considers the latter.
It is thus natural to define the \emph{tripartite dependence} as the situation where any variable depends on all the other variables.
This can be quantified as the worst case conditional mutual information
\begin{eqnarray}
\mathcal{D}_3 & \equiv & \min [ I(X_1 : X_2 | X_3), I(X_1 : X_3 | X_2), \nonumber \\
& & \qquad \quad I(X_2 : X_3 | X_1)  ].
\label{EQ_D3}
\end{eqnarray}
Due to strong subadditivity the conditional mutual information is non-negative and hence $\mathcal{D}_3 \ge 0$~\cite{NielsenChuang}.
$\mathcal{D}_3$ vanishes if and only if there exists a variable such that already a subset of the remaining parties
{\color{black} contains all available information about it.}
Note that this condition is also satisfied if a variable is not correlated with the rest of the system at all.

The value of $\mathcal{D}_3$ can be interpreted using an alternative expression for conditional mutual information, e.g., $I(X_1 : X_3 | X_2) = I(X_1 : X_2 X_3) - I(X_1 : X_2)$.
Accordingly, one recognizes from Eq.~(\ref{EQ_D3}) that $\mathcal{D}_3$ expresses the gain in information about the first subsystem that the second party has from \emph{cooperating} with the third party.
Non-zero value of $\mathcal{D}_3$ ensures that any two parties always gain through cooperation when accessing the knowledge about the remaining subsystem.
The minimal gain over the choice of parties is an alternative way to compute $\mathcal{D}_3$.
{\color{black} Of course by taking this minimum some information about the underlying quantum state is inevitably lost, but our quantity is designed to be sensitive to genuinely multipartite properties of the state only.}

In the context of quantum subsystems we can rewrite the  conditional mutual information as $I(X_1 : X_3 | X_2) = S(X_1 | X_2) + S(X_3 | X_2) - S(X_1 X_3 | X_2)$. 
Since $S(X_1|X_2)$ is the entanglement
cost of merging a state $X_1$ with $X_2$, see Ref.~\cite{Horodecki2005}, we can interpret the conditional mutual information as the extra cost
of merging states one by one ($X_1$ with $X_2$ and $X_3$ with
$X_2$) instead of altogether ($X_1 X_3$ with $X_2$). $\mathcal{D}_3$ is the
minimum extra cost of this merging.

As another example, quantum conditional mutual information prominently appears in the definition of Markov chains.
For a tripartite system in state $\rho_{123}$, vanishing $I(X_1:X_3|X_2) = 0$ means that there exists a recovery map $\mathcal{R}_{X_2 \to X_2 X_3}$, such that $\rho_{123} = \mathcal{R}_{X_2 \to X_2 X_3}(\rho_{12})$, i.e. the global state can be recovered by applying map from $X_2$ to $X_2 X_3$ on marginal $X_1 X_2$. The global state is then called a quantum Markov chain.
More generally, it turns out that the state $\mathcal{R}_{X_2 \to X_2 X_3}(\rho_{12})$ has a high fidelity with $\rho_{123}$ for all states with small $I(X_1:X_3|X_2)$~\cite{Fawzi2015}.
A small values of $\mathcal{D}_3$ therefore indicates that there exists a subsystem from which the global state can be recovered, approximating a quantum Markov chain, and a large $\mathcal{D}_3$ shows that there is no such subsystem.

Moving on to larger systems, we note that there are more conditions to be considered already in order to define the four-partite dependence.
In analogy to the tripartite case the first condition is to require that cooperation of any triple of parties provides more information about the remaining subsystem, e.g., $I(X_1 : X_2 X_3 X_4) - I(X_1 : X_2 X_3)$ must be positive.
But one should also impose that cooperation between any pair brings information gain about the two remaining variables, e.g., $I(X_1 X_2 : X_3 X_4) - I(X_1 X_2 : X_3)$ must be positive.
The former condition demands a positive conditional mutual information, $I(X_1 : X_4 | X_2 X_3) > 0$, while the latter one requires $I( X_1 X_2 : X_4 | X_3)> 0$.
In order to compute $\mathcal{D}_4$ one takes the minimum of these two conditional mutual informations over all permutations of subsystems.
Note, however, that {\color{black} from the chain rule for mutual information and its non-negativity we have, e.g., $I( X_1 X_2 : X_4 | X_3) = I(X_2 : X_4|X_3) + I(X_1 : X_4 | X_2 X_3) \ge I(X_1 : X_4 | X_2 X_3)$ }and therefore it is sufficient to minimize over the conditional mutual information between two variables only.
We emphasize that this step simplifies the computation significantly. 
The same argument applies for arbitrary $N$ and leads to the  definition of $N$-partite dependence
\begin{eqnarray}
\mathcal{D}_N & \equiv & \min_{\textrm{perm}} I(X_1 : X_2 | X_3 \dots X_N),
\label{EQ_DN}
\end{eqnarray}
where the minimum is taken over all permutations of the subsystems.
In the case of a quantum system in state $\rho$ we obtain
\begin{equation}
\mathcal{D}_N (\rho)  =  \min_{i,j} [ S(\mathrm{Tr}_i \rho) + S(\mathrm{Tr}_j \rho)  - S(\mathrm{Tr}_{ij} \rho) - S(\rho)], 
\label{EQ_DN_RHO}
\end{equation}
{\color{black} where $i, j = 1 \dots N$ and $i \ne j$.
$\mathrm{Tr}_i \rho$ denotes a partial trace over the subsystem $i$.}
In general, calculating the $N$-partite dependence requires computation and comparison of ${N \choose 2}$ values, i.e., scales polynomially as $N^2$, whereas for permutationally invariant systems it is straightforward.

One may also like to study $k$-partite dependencies within an $N$-partite system.
To this aim we propose to apply the definitions above to any $k$-partite subsystem and take the minimum over the resulting values.


\section{Properties}

{\color{black} We now prove essential properties of the introduced dependence measure and explain why it is distinct from the multipartite correlations.}


\subsection{Pure states}

First of all, for pure quantum states $| \Psi \rangle$, the dependence can be simplified as
\begin{eqnarray}
\mathcal{D}_N(| \Psi \rangle)  &=& \min_{i,j} [ S(\textrm{Tr}_i | \Psi \rangle \langle \Psi|) \label{pure}\\
&+&S(\textrm{Tr}_j | \Psi \rangle \langle \Psi|)-S(\textrm{Tr}_{ij} | \Psi \rangle \langle \Psi|) ] \nonumber \\
&=& \min_{i,j} [ S(\rho_i)+S(\rho_j)-S(\rho_{ij}) ], \nonumber 
\end{eqnarray}
where $\rho_i$ is the state of the system after removing all but the $i$-th particle, i.e., $\mathcal{D}_N(| \Psi \rangle)$ is given by the smallest quantum mutual information in two-partite subsystems, {\color{black} without any conditioning.}
Here, we made use of the fact that both subsystems of a pure state have the same entropy: $S(\textrm{Tr}_i \rho)=S(\rho_i)$ for $\rho=| \Psi \rangle \langle \Psi|$.
In Appendix~\ref{ogr2} we prove the following upper bound on $\mathcal{D}_N$ for pure states
\begin{equation}
\mathcal{D}_N(| \Psi \rangle) \le 1.
\end{equation}
It is a consequence of the trade-off relation between the quantum mutual information for different two-particle subsystems of a pure global state and the definition of $\mathcal{D}_N$ where the smallest conditional mutual information is chosen.
In particular, the bound is achieved by the $N$-qu$d$it GHZ state $\frac{1}{\sqrt{d}}(\ket{0 \dots 0} + \dots + \ket{d-1 \dots d-1})$.
{\color{black} This can be seen from Eq.~(\ref{pure}).
The two-particle subsystems of the GHZ state are of the form $\rho_{12} = \sum_{j = 0}^{d-1} \frac{1}{d} | j j \rangle \langle jj |$ and their mutual information equals $1$.}
Additionally, the quantum mutual information is bounded by $1$ whenever the state $\rho_{ij}$ is separable.
{\color{black} This follows from the fact that separable states have non-negative quantum conditional entropy $S_{i|j}(\rho_{ij}) \ge 0$~\cite{CerfAdami1999}, and accordingly their mutual information is bounded as $I_{i:j}(\rho_{ij}) = S(\rho_i) - S_{i|j}(\rho_{ij}) \le S(\rho_i) \le 1$, where the last inequality is the consequence of using logarithms to base $d$.}
A comprehensive list of dependencies within standard classes of quantum states is given in Tab. \ref{TAB_TH_STATES}. 
The analytical formula for the $N$-qubit Dicke states with $e$ excitations, $| D_N^e \rangle$, is presented in Appendix \ref{formDicke}. 
In short, if one fixes $e$ and takes the limit $N \to \infty$, the dependence $\mathcal{D}_N$ vanishes. 
For $e$ being a function of $N$, e.g., $e = N/2$, the dependence $\mathcal{D}_N$ tends to $1/2$.

\begin{table}
	\begin{tabular}{c c c c c c}\hline\hline
		$N$ & state & $\mathcal{D}_3$ &$\mathcal{D}_4$ &$\mathcal{D}_5$ &$\mathcal{D}_6$  \\
		\hline 
		3 & $\{P_{\textrm{same}} \}$ & 0  & - & - & - \\
		3 & $\{P_{\textrm{even}} \}$ & 1  & - & - & - \\
		3 & GHZ & 1  & - & - & - \\
	 	3 &$D_{3}^{1}$& 0.9183 & - & - & - \\ 
		3 & $D_{3}^{2}$& 0.9183 & - & - & - \\ 
		3 &$\rho_{\textrm{nc},3}$ & 0.5033 & - & - & -   \\ \hline
		4 & GHZ & 0  & 1 & - & - \\
		4 &$D_{4}^{1}$& 0.3774 & 0.62256 & - & -  \\ 
		4 &$D_{4}^{2}$& 0.5033 & 0.7484 &  - & - \\ 
		4 & $\Psi_4$ & 0.4150 & 0.4150 & - & -  \\
		4 &$L_{4}$& 1 & 0 & - & - \\ 
		4 &3-uniform & 0 & 2 & - & -  \\\hline
		5 & GHZ & 0  & 0 & 1 & - \\
		5 &$D_{5}^{1}$& 0.2490 & 0.2490 & 0.4729 & - \\ 
		5 &$D_{5}^{2}$& 0.3245 & 0.3245 & 0.6464 & - \\ 
		5 &$D_{5}^{3}$ & 0.3245 & 0.3245 & 0.6464 & -   \\ 
		5 & $\rho_{\textrm{nc},5}$ & 0.1710 & 0.6490 & 0.4729 & -   \\ 
		5 &$L_{5}$& 0 & 0 & 0 & - \\ 
		5 &$R_{5}$& 1 & 1 & 0 & - \\ 
		5 &AME(5,2) & 1 & 1 & 0 & -  \\ \hline
		6 & GHZ & 0  & 0 & 0 & 1 \\
		6 &$D_{6}^{1}$& 0.1866 & 0.1634 & 0.1866 & 0.3818 \\ 
		6 &$D_{6}^{2}$& 0.2566 & 0.1961 & 0.2566 & 0.5637 \\ 
		6 &$D_{6}^{3}$& 0.2729 & 0.1961 & 0.2729 & 0.6291  \\ 
		6 &$L_{6}$& 0 & 0 & 0 & 0 \\ 
		6 &$R_{6}$& 0 & 0 & 0 & 0 \\ 
		6 &AME(6,2) & 0 & 2 & 0 & 0 \\ 
		6 &5-uniform & 0 & 0 & 0 & 2 \\ \hline\hline
	\end{tabular}
	\caption{\label{TAB_TH_STATES} Values of the dependence for several quantum states and probability distributions.
		$\{P_{\textrm{same}} \}$ stands for $P(000) = P(111) = \frac{1}{2}$ and $\{P_{\textrm{even}} \}$ for $P(000) = P(110) = P(101) = P(011) = \frac{1}{4}$.
		$D^k_N$ denotes the $N$-partite Dicke states with $k$ excitations $\sim \ket{1\dots 1 0 \dots 0} + \dots + \ket{0 \dots 0 1 \dots 1}$, with $k$ ones;
		$\rho_{\textrm{nc}}$ denotes the genuinely multipartite entangled state without multipartite correlations~\cite{Kaszlikowski2008};
		the GHZ state is described in the text;
		{\color{black} $L_{k}$ and $R_{k}$ stands for the linear cluster and the ring cluster states of $k$ qubits (in general, the graph states are defined by the elements of the stabilizer group for a particular linear or ring graphs, as shown in~\cite{PhysRevA.82.012337})} and $\Psi_4$ is discussed in~\cite{psi4}.
		$k$-uniform states are states where all $k$-partite marginals are maximally mixed, whereas AME(n,d), so-called absolutely maximally entangled states, refer to $\lfloor n/2 \rfloor$-uniform states of $d$ dimensions~\cite{Helwig2012}.
	}
\end{table}



\subsection{Maximum dependence}

The maximal $N$-partite dependence over classical distributions of $d$-valued variables is given by $1$ (recall that our logarithms are base $d$)
and follows from the fact that classical mutual information cannot exceed the entropy of each variable.
On the other hand, quantum mutual information is bounded by $2$ and this is the bound on $\mathcal{D}_N$ optimized over quantum states (see Appendix~\ref{APP_MAX_STATE}).
{\color{black}As seen, pure quantum states satisfy the classical bound of $1$}, but there exist \textit{mixed} states belonging to the class of $k$-uniform states, in particular for $k = N-1$~\cite{kuniform}, achieving the bound of $2$. In the case of $N$ qubits (for $N$ even) the optimal states have the following form
\begin{eqnarray}
\rho_{\max} & = & \frac{1}{2^N} \left( \sigma_0^{\otimes N} + (-1)^{N/2} \sum_{j = 1}^3 \sigma_j^{\otimes N}  \right),
\label{EQ_RHO_MAX}
\end{eqnarray}
where $\sigma_j$ are the Pauli matrices and $\sigma_0$ denotes the $2 \times 2$ identity matrix.
Note that $\rho_{\max}$ is permutationally invariant
and gives rise to perfect correlations or anti-correlations when all observers measure locally the same Pauli observable.
These states are known as the generalized bound entangled Smolin states \cite{smolin1,smolin2}. 
They are a useful quantum resource for multiparty communication schemes~\cite{Remik2} and were experimentally demonstrated in Refs.~\cite{smolin3,Lavoie2010comment,Amselem2010,Lavoie2010,Barreiro2010,Amselem2013}.
Per definition for $(N-1)$-uniform states all reduced density matrices are maximally mixed, with vanishing mutual information, whereas the whole system is correlated.
In Appendix~\ref{APP_MAX_STATE} we provide examples of states which maximize $\mathcal{D}_N$ for arbitrary $d$ and show in general that the only states achieving the maximal quantum value of $2$ are $(N-1)$-uniform.

Let us also offer an intuition for values of $\mathcal{D}_N$ above the classical bound of one.
As shown in Appendix~\ref{ogr2} this can only happen for mixed quantum states.
One could then consider an auxiliary system which purifies the mixed state.
High values of $\mathcal{D}_N$ correspond to learning simultaneously the variables of the subsystems and the auxiliary system.
Note that making this statement mathematically precise may be difficult as the problem is equivalent to the interpretation of negative values of conditional entropy~\cite{Horodecki2005,Rio2011,Chuan2012}.

A practical implication of the fact that mixed states achieve maximal dependence is that coupling an initially closed system to an environment can improve multipartite dependencies within the system.
Furthermore, this is the only way of generating non-classical values of $\mathcal{D}_N$ as dynamics within a closed system initialised in a pure state cannot beat the classical bound.

\subsection{Comparison with multipartite correlations}

Let us begin with a simple example that illustrates the difference between multipartite correlations and multipartite dependence.
Consider three classical binary random variables described by the joint probability distribution $P(000) = P(111) = \frac{1}{2}$.
All three variables are clearly correlated as confirmed, e.g., by quantifiers introduced in Refs.~\cite{Giorgi,Girolami2017}.
However, the knowledge of, say, the first party about the third party does not increase if the first observer is allowed to cooperate with the second one.
By examining her data, the first observer knows the variables of both remaining parties and any cooperation with one of them does not change this.
There is no information gain and hence this distribution has vanishing tripartite dependence.

On the other hand, let us consider the joint probability distribution with $P(000) = P(011) = P(101) = P(110) = \frac{1}{4}$, which can describe also a classical system.
Any two variables in this distribution are completely uncorrelated, but any two parties can perfectly decode the value of the remaining variable.
Hence the gain from cooperation is $1$ and so is the value of $\mathcal{D}_3$.

Nevertheless, $\mathcal{D}_N$ does satisfy a number of properties that are expected from measures of genuine multipartite correlations.
Any such quantifier should satisfy a set of postulates put forward in Refs.~\cite{Grudka,Girolami2017}.
We now show that most of them also hold for $\mathcal{D}_N$ and we precisely characterize the deviation from one of the postulates.
In Appendices A-C we prove the following properties of the dependence:
\begin{itemize}

\item[(i)] If $\mathcal{D}_N = 0$ and one adds a party in a product state then the resulting $(N+1)$-party state has $\mathcal{D}_{N} = 0$.

\item[(ii)] If $\mathcal{D}_N = 0$ and one subsystem is split with two of its parts placed in different laboratories then the resulting $(N+1)$-party state has $\mathcal{D}_{N+1} = 0$.

\item[(iii)] $\mathcal{D}_N$ can increase under local operations. Let us denote with the bar the quantities computed after local operations. We have the following inequality:
\begin{eqnarray}
\label{EQ_D_NOMONO}
\overline{\mathcal{D}}_N & \le & \mathcal{D}_N + I(X_1 X_2 : X_3 \dots X_N) \nonumber \\
&& \qquad - I(X_1 X_2 : \overline X_3 \dots \overline X_N),
\end{eqnarray}
where systems $X_1$ and $X_2$ are the ones minimizing $\mathcal{D}_N$, i.e., before the operations were applied.

\end{itemize}

The properties (i) and (ii) hold for all quantifiers of multipartite correlations.
It is expected that measures of multipartite correlations are also monotonic under local operations (though note that often this condition is relaxed in practice, see e.g. {\color{black}quantum discord~\cite{Ollivier2001,Modi2010,serra2011})}. In the present case, the monotonicity property does not hold in general for $\mathcal{D}_N$, however, property (iii) puts a bound on its maximal violation {\color{black} (see Appendix~\ref{APP_LOUP} for a concrete example)}.
Moreover, it has a clear interpretation: local operations that uncorrelate a given subsystem from the others may lead to information gain when the less correlated party cooperates with other parties.

Let us explain this more quantitatively for the conditional mutual information between variables $X_1$ and $X_2$.
While it is well-known that this quantity is monotonic under local operations on subsystems not in the condition~\cite{Wilde2018},
we prove in Appendix~\ref{APP_LOCI} that the following inequality is satisfied under local operations on arbitrary subsystem (being the origin of property (iii)):
\begin{eqnarray}
&& I(\overline{X}_1 : \overline{X}_2 | \overline{X}_3 \dots \overline{X}_N) \le I(X_1 : X_2 | X_3 \dots X_N) \nonumber \\
& + & I(X_1 X_2 : X_3 \dots X_N) - I( X_1 X_2 : \overline{X}_3 \dots \overline{X}_N).
\end{eqnarray}
The second line is non-negative due to the data processing inequality and it quantifies how much the local operations have uncorrelated the variables in the condition $X_3 \dots X_N$ from the variables $X_1 X_2$.
This sets the upper bound to the lack of monotonicity of the conditional mutual information.

{\color{black}

Let us also mention that the lack of monotonicity under local operations has been discussed in the context of complexity measures for multipartite systems~\cite{Zhou2009,KOJA2009,GG2012}. Those measures involve a different concept of correlation, in terms of the number of particles that have to be coupled in the Hamiltonian for which the discussed state is a thermal state. However, these measures a similar spirit to the dependence in that they capture the improvement in approximating a given distribution when more and more particles are coupled in the Hamiltonian.

\subsection{Partial extension of classical interpretations}
As shown, in classical information theory $\mathcal{D}_N$ has a clear interpretation as information gain from cooperation.
Agents could measure their subsystems and the information gain from the outcomes would match the computed value of $\mathcal{D}_N$.
In the following we partially extend this interpretation to quantum systems, when the values of $\mathcal{D}_N$ do not exceed unity.
Let us proceed by comparing examples of classical and quantum distributions.
As shown in the previous section a classical mixture of $P(000) = P(111) = \frac{1}{2}$ admits $\mathcal{D}_3 = 0$.
To the contrary, a superposition $\frac{1}{\sqrt{2}}(| 000 \rangle + | 111 \rangle)$ admits $\mathcal{D}_3 = 1$.
The quantum coherence here improves the dependence because we can measure the GHZ state in a different basis and the classical dataset obtained gives rise to $\mathcal{D}_3 = 1$.
The relevant basis is to measure each qubit along $\frac{1}{\sqrt{2}}(|0 \rangle \pm | 1 \rangle)$ directions.

As another example consider the classical distribution $P(100) = P(010) = P(001) =  \frac{1}{3}$ which gives rise to $\mathcal{D}_3 = \frac{2}{3}$.
The corresponding quantum superposition $| W \rangle = \frac{1}{\sqrt{3}}(| 100 \rangle + | 010 \rangle + |001 \rangle )$ has a rather non-trivial value of dependence given by $\mathcal{D}_3 = 0.9183$.
However, this does not imply that there exists a set of local measurements on $| W \rangle$ which yields a classical distribution with higher dependence. In fact, by optimising the dependence of local measurement results over all projective quantum measurements one obtains the classical $\frac{2}{3}$.
Note, however, that $\mathcal{D}_3$ compares the information two systems have about the third one with the pairwise information.
It is therefore natural to also consider joint measurements on two parties.
We have therefore computed the post-measurement state 
$\sum_{i,j} \Pi_j^{12} \otimes \Pi_k^{3} | W \rangle \langle W | \Pi_j^{12} \otimes \Pi_k^{3}$,
where $\Pi_j^{12}$ are the rank-one projectors on the first two particles and correspondingly $\Pi_k^{3}$ are for the last qubit.
Indeed, when we optimise over these projectors the dependence of the post-measurement state precisely matches the value computed for the W state, i.e. $\mathcal{D}_3 = 0.9183$.

The examples given suggest that perhaps the interpretation of the quantum value of $\mathcal{D}_N$ (whenever not exceeding unity) can be given as the highest dependence of the classical dataset that can be measured on the quantum state (including joint measurements).
It turns out that this is in general not the case.
We have found examples of pure four-qubit states for which there exist local measurements with outcomes producing $\mathcal{D}_4$ higher than that of the corresponding quantum states. 
This further demonstrates property (iii) listed above.

}


\section{Applications and examples}

Multipartite dependence can be computed for both classical and quantum systems 
and is a generic quantifier of information gain from cooperation that can be used across science.
Here we discuss examples of applications of $\mathcal{D}_N$ in quantum information {\color{black} and briefly mention its role in data science.}


\subsection{Quantum secret sharing}

An intuitive application of $\mathcal{D}_N$ is secret sharing~\cite{Shamir79,Blakley79,Hillery1999,Imai2005}
{\color{black} with the additional constraint that the secret could be shared by any party. We refer to this problem as {\em symmetric secret sharing}.}
In the tripartite setting it requires collaboration of two parties in order to read out the secret of the remaining party. 
In the classical version of this problem the secret is a random variable, e.g., the measurement outcome of, say, the first observer.
It is thus required that both, the second as well as the third party alone has only little or no information about the secret, i.e., $I(X_1 : X_2)$ and $I(X_1 : X_3)$ are small,
while both of them together can reveal the result of the first observer, i.e., $I(X_1 : X_2 X_3)$ is large or unity.
Clearly, $\mathcal{D}_3$ is the relevant figure of merit and due to the minimization in (\ref{EQ_D3}), the secret can be generated at any party.
The states $\rho_{\max}$ derived above appear well suited for this task and since they admit perfect correlations along complementary local measurements, by following the protocol in~\cite{Hillery1999}, the quantum solution to the secret sharing problem offers additionally security against eavesdropping.

{\color{black} This security has also been recently quantified with the conditional quantum mutual information in Ref.~\cite{Sharma2020}. The dependence introduced here is in contrast taking into account that any participant could be sharing the secret.
	
	Let us also define the \emph{quantum} secret sharing task and demonstrate the relevance of the dependence in this context.
	Suppose Alice has a quantum state $\rho$, called the secret, which she wants to split into $n$ shares such that the secret is recoverable only when a party has all $n$ shares.
	A quantum secret sharing scheme \cite{Imai2005} is a map $\mathcal{E}_n: A \to X^{\otimes n}$ such that,
	\begin{equation}
	C_Q(\textrm{Tr}_k \circ \mathcal{E}_n) = 0
	\end{equation}
	where $\textrm{Tr}_k$ is the partial trace over an arbitrary set of subsystems and $C_Q(\Lambda)$ is the quantum capacity of the channel $\Lambda$.
	This condition encodes the requirement that from any subset of shares one is not capable of recovering the initial quantum information in $\rho$.
	The rate of a quantum secret sharing scheme is given by the quantum capacity of the channel $\mathcal{E}_n$.
	
	The concrete protocol utilizing $\rho_{\max}$ is as follows.
	Consider a quantum secret $\rho$ of a single qubit.
	Using the teleportation protocol with state $\rho$ and one subsystem of $\rho_{\max}$ in Alice's possession, she performs the
	encoding map $\mathcal{E}_{N-1}$ to her qubit:
	\begin{equation}
	\mathcal{E}_{N-1} (\rho) = \frac{1}{2^{N-1}} \left( \sigma_0^{\otimes N-1} + (-1)^{N/2} \sum_{j = 1}^3 \sigma_j^{\otimes N-1} \textrm{Tr} \left( \sigma_j^T \rho \right)  \right).
	\end{equation}
	In this way Alice prepares $n = N-1$ shares of the secret. We denote the state above as $\rho_{N-1}$.
	Since for any $\rho$ we have $(\textrm{Tr}_k \circ \mathcal{E}_{N-1}) (\rho) \propto \identity$, it follows that $C_Q(\textrm{Tr}_k \circ \mathcal{E}_{N-1}) = 0$, i.e., no subset of observers can recover the quantum secret.
	All of them together, however, can recover it perfectly with the following ``reverse'' teleportation scheme.
	Again take the resource $(N-1)$-uniform state
	and conduct a measurement in the basis $\{(\openone \otimes \dots \otimes \openone \otimes \sigma_{\mu_1} \otimes \dots \otimes \sigma_{\mu_{N-1}}) | \Psi \rangle \}$,
	where $\mu_n = 0,1,2,3$ and 
	\begin{equation}
	| \Psi \rangle = \frac{1}{\sqrt{2^{N-1}}} \sum_{j_1 \dots j_{N-1} = 0,1} | j_1 \dots j_{N-1} \rangle \otimes | j_1 \dots j_{N-1} \rangle
	\end{equation}
	is the maximally entangled state.
	This measurement is conducted on the shares $\rho_{N-1}$ and $N-1$ subsystems of $\rho_{\max}$.
	If the measurement result corresponds to $| \Psi \rangle$ the remaining qubit is in the state $\rho$, otherwise there exist unitary operations depending on the result that transform the single qubit output to $\rho$.
	
	We now show that any $N$-partite state $\rho_c$ with maximally mixed marginals and non-classical dependence $\mathcal{D}_{N}(\rho_c) > 1$ is useful for the quantum secret sharing.
	Consider the encoding map $\mathcal{E}_{c}: A \to X^{\otimes N-1}$
	with the Choi state given by $\rho_c$~\cite{Choi1975}, i.e., $(\openone \otimes \mathcal{E}_c)(|\Phi \rangle \langle \Phi |) = \rho_c$, where $| \Phi \rangle$ is the maximally entangled state.
	The rate of quantum secret sharing admits the lower bound
	\begin{subequations}
		\begin{align}
		R
		&= C_Q ( \mathcal{E}_{c} ) \label{eq:ssr2} \\
		& \ge  \sup_{\phi_{A' A}} - S_{A'|X_1 ... X_{N-1}}\left( ( \identity \otimes \mathcal{E}_{c} ) (\phi_{A' A}) \right) \label{eq:ssr3} \\
		&\geq - S_{A'|X_1 ... X_{N-1}} (\rho_c) \label{eq:ssr4} \\
		&= I(A':X_{1} | X_2 ... X_{N-1}) - S(A'|X_{2} ... X_{N-1}) \label{eq:ssr5} \\
		&= I(A':X_{1} | X_2 ... X_{N-1}) - 1 \label{eq:ssr7} \\
		& \geq  \mathcal{D}_{N}(\rho_c) - 1. \label{EQ_TIGHT}
		\end{align}
	\end{subequations}
	The steps are justified as follows.
	The first line follows from definition.
	Ineq. (\ref{eq:ssr3}) is the result of computing the quantum capacity of a channel \cite{Schumacher1996,Schumacher1996a,Barnum1998,Barnum2000,Lloyd1997,Devetak2003} with system $A'$ being of the same dimension as system $A$ {\color{black} and where $S_{A'|X_1 ... X_{N-1}} (\rho)$ is the quantum conditional entropy of state $\rho$},
	(\ref{eq:ssr4}) follows because the maximally entangled state is a particular choice of $\phi_{A'A}$, and the Choi state of $\mathcal{E}_{c}$ is $\rho_c$.
	Eq. (\ref{eq:ssr5}) follows from the properties of entropy recalling that our logarithms are base $d$.
	Eq. (\ref{eq:ssr7}) follows from the assumption that the state has maximally mixed marginals.
	Finally, the dependence is the worst case conditional mutual information.
	
	Since the marginals of $\rho_c$ are maximally mixed, the same holds for the encoded state $\rho_{N-1} = \mathcal{E}_c(\rho)$, i.e., no subset of parties can recover the quantum secret alone,
	yet for all of them together $R > 0$ holds for $\mathcal{D}_{N}(\rho_c) > 1$.
	
	This lower bound is in general not tight, although it is achieved, e.g., by the $(N-1)$-uniform states~\cite{Helwig2012}.
	In fact, all degradable channels give rise to the equality in (\ref{eq:ssr3}) and all symmetric states admit the equality in (\ref{EQ_TIGHT}).
	Note that for all pure states the lower bound on the rate is zero, whereas, e.g., \cite{Hillery1999} gives a quantum secret sharing scheme using a GHZ state with a unit rate.
	In this particular case it is easy to generalize the proof above.
	Since the GHZ state has classically correlated marginals (not maximally mixed), the conditional entropy in Eq. (\ref{eq:ssr5}) vanishes and the rate is lower bounded by the dependence alone, which is $1$ for the GHZ state.
	
	Finally, note that also situations where a subset of $k$ parties is required to read the secret are of practical interest. Analogical arguments to the ones just given show that the dependence $\mathcal{D}_k$ is the figure of merit for sharing the secret among any $k$-partite subsystem of $N$-party state where anyone could be the secret sharer.
}


{\color{black}
	\subsection{Witnessing entanglement}
	
	As derived above, the values of $\mathcal{D}_N$ exceeding $1$ indicate that quantum systems are being measured.
	Furthermore, the quantum state of the systems is not pure.
	We now show that such values witness quantum entanglement.
	
{\color{black} By our assumption $\mathcal{D}_N > 1$.
Per definition of dependence, it is the smallest difference of mutual informations.
Let us label the subsystems such that this minimum is
$\mathcal{D}_N = I(X_1:X_2 X_3 \dots X_N) - I(X_1:X_3 \dots X_N)$,
Since the second term is never positive, it is the first mutual information that has to be bigger than one, i.e. $I(X_1:X_2 X_3 \dots X_N) > 1$.
Writing the mutual information in terms of the quantum conditional entropy then gives
	\begin{equation}
	S_{X_1 | X_2 X_3 \dots X_N}(\rho) < -1 + S(\rho_1),
	\end{equation}
	For subsystems with the same dimension $S(\rho_1) \le 1$ and accordingly
	the conditional mutual information is negative.
	As shown by Cerf and Adami, this is only possible for entangled states~\cite{CerfAdami1999}.}
	Note that this entanglement does not have to be genuinely multipartite.
	{\color{black} An example of a particular state violating the bound of $1$ is given in Eq.~(\ref{EQ_RHO_MAX}).
	It also nicely demonstrates the point just given as it can be written as a mixture of correlated Bell states~\cite{smolin1} and therefore is biseparable. 
Furthermore, the proof can be repeated for any number of subsystems $k$, i.e. if $\mathcal{D}_k > 1$ any $k$-party subsystem is entangled. A concrete example where this is the case is given by the four-party subsystems of absolutely maximally entangled state of $6$ qubits (see Tab.~\ref{TAB_TH_STATES}).	
	}
}


{\color{black}
\subsection{Entanglement without dependence}
}

An intriguing question in the theory of multipartite entanglement is whether entanglement can exist without classical multipartite correlations~\cite{Kaszlikowski2008}.
The examples of $N$-party entangled states with vanishing $N$-party classical correlations are known in the literature~\cite{Lask2012,Schwemmer2015,Designolle2017,Tran2017,Klobus2019}, though the corresponding notions of classical correlations do not satisfy all the postulates of Refs.~\cite{Grudka,Girolami2017}.
Here we ask whether there are genuinely multipartite entangled states with no multipartite dependence.

It turns out there exist even pure genuinely multipartite entangled states without multipartite dependence.
Consider any $N$-qudit cluster state (including linear, ring, 2D, etc.) for $N \geq 4$.
It was shown in Ref.~\cite{PhysRevA.82.012337} that all single-particle subsystems are completely mixed and there exists at least one pair of subsystems in the bipartite completely mixed state.
The corresponding entropies are equal to $S(\rho_i) = 1$ and $S(\rho_{ij})=2$, and lead to $\mathcal{D}_N = 0$, due to Eq. (\ref{pure}). 
Therefore, the information about a particular subsystem cannot be increased when other subsystems are brought together which explains the impossibility of the corresponding secret sharing task~\cite{cluster-secret-sharing,cluster-secret-sharing-qudit,cluster-secret-sharing-erratum}.
Note that there exist other subsets of observers who can successfully run secret sharing using a cluster state.

{\color{black}
\subsection{Dependence without correlations}
}

Similarly we ask whether multipartite dependence can exist without multipartite correlations and vice versa.
It turns out that all combinations are possible.
The cluster states discussed in the previous subsection give rise to multipartite correlations and therefore show that multipartite correlations can exist without multipartite dependence.

Conversely, the dependence can be non-zero even in states with no correlations whatsoever.
To this end consider the state $\rho_{\mathrm{nc}} = \frac{1}{2} | D_N^1 \rangle \langle D_N^1| + \frac{1}{2} | D_N^{N-1} \rangle \langle D_N^{N-1} |$. 
It is $N$-party entangled and has vanishing all $N$-partite correlation functions~\cite{Kaszlikowski2008}. 
Yet, its $\mathcal{D}_N$ is finite as shown in Tab.~\ref{TAB_TH_STATES}.
This again shows that multipartite dependence is distinct from multipartite correlations and captures other properties of genuinely multi-partite entangled systems.


\subsection{Experimental states}

We move to multipartite dependence in quantum optics experiments.
{\color{black} Table~\ref{tabela} gathers quantum states prepared with photonic qubits. 
	Details of the experimental preparation of the states and the evaluation of the dependence are given in Appendix~\ref{APP_EXP}.}
We have chosen to present the states illustrating the properties discussed above.

The experimental data is in good agreement with the theoretical calculations.
Deviations for the six qubit state $D_{6}^{3}$ result from reduced fidelities due to contributions of higher order noise in the state preparation. 
The same applies to the five qubit state $\rho_{\textrm{nc},5}$ derived from $D_{6}^{3}$.
Indeed, the states denoted as $\rho_{\mathrm{nc}}$, which have vanishing correlation functions between all $N$ observers~\cite{Schwemmer2015}, clearly show a non-vanishing value for $\mathcal{D}_N$.
Hence, these states are examples for ``entanglement without correlations'' and ``dependence without correlations''.
Similarly, the experimental data of the linear cluster state $L_4$ indicates ``entanglement without dependence'' and ``correlations without dependence''.
In the experiment, the GHZ$_4$ state $\sim \ket{0000} + \ket{1111}$ achieves the highest dependence of all considered states and is close to the theoretical dependence $\mathcal{D}_4=1$, which is maximal over all pure states.
The small value of $\mathcal{D}_3$ for the four-partite GHZ state reflects its property of having vanishing dependence for all tripartite classically correlated subsystems.

\begin{table}
	\begin{tabular}{c c c c c c} \hline \hline
		$N$ & state & $\mathcal{D}_3$ &$\mathcal{D}_4$ &$\mathcal{D}_5$ &$\mathcal{D}_6$  \\ \hline
		3 & $D_{3}^{1}$& 0.79 (0.92) & - & - & - \\ 
		3 & $D_{3}^{2}$& 0.82 (0.92) & - & - & - \\ 
		3 & $\rho_{\textrm{nc},3}$ & 0.44 (0.50) & - & - & -   \\  \hline
		4 & GHZ$_4$ & 0.06 (0.00) & 0.95 (1.00) & - & - \\
		4 & $D_{4}^{2}$ & 0.41 (0.50) & 0.66 (0.75) &  - & - \\ 
		4 & $L_{4}$ & 0.90 (1.00) & 0.09 (0.00) & - & - \\ 
		4 & $\Psi_4$ & 0.33 (0.42) & 0.39 (0.42) & - & -  \\ \hline
		5 & $D_{5}^{2}$ & 0.22  (0.32) & 0.17 (0.32) & 0.21 (0.65) & -   \\ 
		5 & $D_{5}^{3}$ & 0.23  (0.32) & 0.19 (0.32) & 0.22 (0.65) & -   \\ 
		5 & $\rho_{\textrm{nc},5}$ & 0.21  (0.17) & 0.15 (0.65) & 0.15 (0.47) & -   \\ \hline
		6 & $D_{6}^{3}$ & 0.21 (0.27) & 0.14 (0.20) &  0.15 (0.27) & 0.19 (0.63)  \\ 
		\hline \hline
	\end{tabular}
	\caption{\label{tabela} Illustrative values of dependence for several experimental quantum states.
		In brackets we give theoretical predictions for ideal states.}
\end{table}

{\color{black}	

	\subsection{Data science}
	
	The dependence is also expected to find applications outside physics and we briefly sketch how it can be useful in data science.
	The problem of feature selection is to reduce available data to a smaller subset that faithfully represents the whole dataset,
	so that the predictions made on the basis of the subset would be the same as based on the entire set.
	In other words, we would like to eliminate variables that are not important.
	
	Such variables can be identified from the conditions of extremal dependence.
	Other functions of conditional mutual information were considered in Refs.~\cite{Yang1999,Brown2012}.
	Let us first analyze the case of $\mathcal{D}_N \approx 0$.
	For example, if the minimizing conditional information is
	$I(X_1 : X_2 | X_3 \dots X_N) \approx 0$ then either variable $X_1$ or $X_2$ can be eliminated as it does not improve the information between the remaining variables,
	e.g., $I(X_1 : X_2 X_3 \dots X_N) \approx I(X_1 : X_3 \dots X_N)$.
	On the other hand, $\mathcal{D}_N \approx 1$ corresponds to the situation where each variable is independent of the rest, e.g., $I(X_1 : X_3 \dots X_N) \approx 0$, (and hence at first sight one would have to keep track of all of them), but from $N-1$ variables one can predict the remaining one, e.g. $I(X_1 : X_2 X_3 \dots X_N) \approx 1$.  Accordingly, one variable can be eliminated.
}


\section{Conclusions}

{\color{black} We have introduced a quantity, the multipartite dependence, as new tool for the characterisation of quantum states. It is surely the method of choice to determine whether and by what amount cooperation between any subsystems brings additional information about the remaining subsystems. It offers an extension to the characterization of multipartite properties via multipartite correlations.
The dependence is directly calculable and has a clear interpretation. As such, several applications are identfied here clearly indicating its relevance for future studies in quantum communication and elsewhere.}

\begin{acknowledgments}
We thank Krzysztof Szczygielski for valuable discussions. 
The work is supported by DFG (Germany) and NCN (Poland) within the joint funding initiative ``Beethoven2'' (2016/23/G/ST2/04273, 381445721), 
by the Singapore Ministry of Education Academic Research Fund Tier 2 Project No. MOE2015-T2-2-034, and by Polish National Agency for Academic Exchange NAWA Project No. PPN/PPO/2018/1/00007/U/00001.
W.L. and R.G. acknowledge partial support by the Foundation for Polish Science (IRAP project, ICTQT, Contract
No. 2018/MAB/5, cofinanced by EU via Smart Growth Operational Programme).
JD and LK acknowledge support from the PhD programs IMPRS-QST and ExQM, respectively.
JDMA is funded by the Deutsche Forschungsgemeinschaft (DFG, German Research Foundation) under Germany's Excellence Strategy - EXC-2111 - 390814868.
\end{acknowledgments}

\appendix

\section{Proof of property (i)}

If $\mathcal{D}_N = 0$ and one adds a party in a product state then the resulting $(N+1)$-partite state has $\mathcal{D}_{N} = 0$.

\begin{proof}
Per definition, we are minimizing the conditional mutual information over all $N$-partite subsystems of the total $(N+1)$-party state.
If one takes the $N$-partite subsystem that excludes the added party, by assumptions $\mathcal{D}_N = 0$.
\end{proof}

In other words, if the cooperation of $N-1$ parties within the $N$-partite system does not help in gaining additional knowledge about any other remaining party, then the cooperation with any additional independent system will not help either.

\section{Proof of property (ii)}

If $\mathcal{D}_N = 0$ and one subsystem is split with two of its parts placed in different laboratories then the resulting $(N+1)$-party state has $\mathcal{D}_{N+1} = 0$.

\begin{proof}
Without loss of generality and in order to simplify notation let us consider an initially tripartite system where the third party is in possession of two variables labeled $X_3$ and $X_4$.
The splitting operation places these variables in separate laboratories producing a four-partite system.
By assumption $\mathcal{D}_3 = 0$, but this does not specify which conditional mutual information in Eq.~(\ref{EQ_D3}) vanishes.
If this is the mutual information where the variables $X_3$ and $X_4$ of the third party enter in the condition, then this mutual information is also minimizing $\mathcal{D}_4$, and hence the latter vanishes.
The second possibility is that the variables of the third party enter outside the condition, e.g., the vanishing conditional mutual information could be $I(X_1: X_3 X_4 | X_2)$.
From the chain rule for mutual information, $0=I(X_1: X_3 X_4 | X_2) \ge I(X_1 : X_4 | X_2 X_3)$.
Finally, from strong subadditivity follows $\mathcal{D}_4 = 0$.
In the $N$-partite case one writes more variables in the conditions and follows the same steps.
\end{proof}

\section{Proof of property (iii)}
\label{APP_LOCI}

Consider a state $\rho$ that is processed by general local operations (CPTP maps) to a state $\overline{\rho}$.
The following upper bound on the multipartite dependence after local operations holds:
\begin{eqnarray}
\overline{\mathcal{D}}_N & \le & \mathcal{D}_N + I(X_1 X_2 : X_3 \dots X_N) \nonumber \\
&& \qquad - I(X_1 X_2 : \overline X_3 \dots \overline X_N),
\end{eqnarray}
where systems $X_1$ and $X_2$ are the ones minimizing $\mathcal{D}_N$, i.e., before the operations were applied.

Let us begin with a lemma characterizing the lack of monotonicity of conditional mutual information under local operations.
\begin{lemma}
\label{LM_MONO}
The following inequality holds:
\begin{eqnarray}
&& I(\overline{X}_1 : \overline{X}_2 | \overline{X}_3 \dots \overline{X}_N) \le I(X_1 : X_2 | X_3 \dots X_N) \nonumber \\
& + & I(X_1 X_2 : X_3 \dots X_N) - I( X_1 X_2 : \overline{X}_3 \dots \overline{X}_N),
\end{eqnarray}
where bars denote subsystems transformed by arbitrary local CPTP maps.
\end{lemma}
\begin{proof}
The conditional mutual information is already known to be monotonic under operations on systems not in the condition~\cite{Wilde2018}:
\begin{eqnarray}
I(\overline{X}_1 : \overline{X}_2 | \overline{X}_3 \dots \overline{X}_N) \le I( X_1 : X_2 | \overline{X}_3 \dots \overline{X}_N)
\end{eqnarray}
Now we continue as follows:
\begin{eqnarray*}
&& I( X_1 : X_2 | \overline{X}_3 \dots \overline{X}_N) + I( X_1 X_2 : \overline{X}_3 \dots \overline{X}_N) \\
& = & I( X_1 : X_2 \overline{X}_3 \dots \overline{X}_N) + I( X_2 : X_1 \overline{X}_3 \dots \overline{X}_N) \\
& & - I(X_1 : X_2)  \\
& \le & I( X_1 : X_2 X_3 \dots X_N) + I( X_2 : X_1 X_3 \dots X_N) \\
&& - I(X_1 : X_2)  \\
& = & I( X_1 : X_2 | X_3 \dots X_N) + I( X_1 X_2 : X_3 \dots X_N),
\end{eqnarray*}
where the first equation is obtained by manipulating entropies such that the mutual informations containing barred subsystems come with positive sign, next we used the data processing inequality and in the last step we reversed the manipulations on entropies.
This completes the proof of the lemma.
\end{proof}

To complete the proof of property (iii) we write 
\begin{eqnarray*}
\mathcal{D}_N & = & I( X_1 : X_2 | X_3 \dots X_N) \\
& \ge & I(\overline{X}_1 : \overline{X}_2 | \overline{X}_3 \dots \overline{X}_N) - I(X_1 X_2 : X_3 \dots X_N) \\
& & \quad + I( X_1 X_2 : \overline{X}_3 \dots \overline{X}_N) \\
& \ge & \overline{\mathcal{D}}_N - I(X_1 X_2 : X_3 \dots X_N) \\
&& + I( X_1 X_2 : \overline{X}_3 \dots \overline{X}_N),
\end{eqnarray*}
where in the first line we denote the subsystems such that the conditional mutual information $I( X_1 : X_2 | X_3 \dots X_N)$ achieves minimum in $\mathcal{D}_N$. Next, the first inequality follows from Lemma~\ref{LM_MONO}, 
and the second inequality from the fact that $I(\overline{X}_1 : \overline{X}_2 | \overline{X}_3 \dots \overline{X}_N)$ may not be the one minimizing $\overline{\mathcal{D}}_N$.

\section{Increasing $\mathcal{D}$ with local operations}
\label{APP_LOUP}

We now give an analytical example where $\mathcal{D}_3$ increases under local operation on the system in the condition.
Consider the following classical state
\begin{eqnarray}
\rho &=& \frac{1}{2}|000\rangle\langle000| + \frac{1}{8}|101\rangle\langle101| \\&+& \frac{1}{8}|110\rangle\langle110| + \frac{1}{4}|111\rangle\langle111|.\nonumber
\end{eqnarray}
One verifies that its $3$-dependence equals $\mathcal{D}_3(\rho) = I(X_2:X_3|X_1)=0.06$, i.e., conditioning on $X_1$ gives the smallest conditional mutual information.
The application of an amplitude-damping channel with Kraus operators
\begin{eqnarray} 
K_0 = 
\left( {\begin{array}{cc}
	0 & 1/\sqrt{2} \\
	0 & 0 \\
	\end{array} } \right), \;\;\;
K_1 = 
\left( {\begin{array}{cc}
	1 & 0 \\
	0 & 1/\sqrt{2} \\
	\end{array} } \right),
\end{eqnarray} 
on subsystem $X_1$ produces the state $\overline{\rho}$, for which one computes $\mathcal{D}_3(\overline{\rho}) = I(\overline{X_1}:X_2|X_3) = I(\overline{X_1}:X_3|X_2)=0.19$. 
Note the change in the conditioned system minimizing the dependence.
The local operation on $X_1$ has increased the information $I(X_2:X_3|\overline{X_1})$ above the other two conditional mutual informations.

\section{Quantum qudit states maximizing $\mathcal{D}_N$}
\label{APP_MAX_STATE}

Let us consider a quantum state of $N$ qu$d$its, for $N$ being a multiple of $d$ and $N \geq 3$, defined as the common eigenstate of the generators
\begin{equation}
G_1^{(d)} = \bigotimes_{i=1}^N X^{(d)}, ~~~ G_2^{(d)}= \bigotimes_{i=1}^N Z^{(d)}, 
\end{equation}
composed of $d$-dimensional Weyl-Heisenberg matrices $X^{(d)}= \sum_{j=0}^{d-1} |j \rangle \langle j+1|$, and $Z^{(d)} = \sum_{j = 0}^{d-1} \omega^j |j \rangle \langle j|$, with $\omega = e^{i 2 \pi/d}$. 
The explicit form of the state can be calculated in the following way:
\begin{equation}
\rho^{(d)}_N= \frac{1}{d^N} \sum_{i,j=0}^{d-1} (G_1^{(d)})^i (G_2^{(d)})^j.
\label{gen}
\end{equation}	
The state (\ref{gen}) belongs to the class of $k$-uniform mixed states defined in \cite{kuniform}, with $k=N-1$.

It is known that for $N$ even the state $\rho^{(d)}_N$ has $d^{N-2}$ eigenvalues equal to $\frac{1}{d^{N-2}}$, so the entropy $S(\rho^{(d)}_N)$ is equal to 
\begin{equation}
S(\rho^{(d)}_N) = N-2.
\end{equation}
Since the state is $(N-1)$-uniform, all reduced density matrices are proportional to identity matrices giving
\begin{eqnarray}
S(\textrm{Tr}_i \rho^{(d)}_N) = N-1, \\
S(\textrm{Tr}_{i,j} \rho^{(d)}_N) = N-2.
\end{eqnarray}
Therefore, for $N$ even
\begin{eqnarray}
\mathcal{D}_N(\rho^{(d)}_N) = S(\textrm{Tr}_i \rho^{(d)}_N) + S(\textrm{Tr}_j \rho^{(d)}_N) \\
- S(\textrm{Tr}_{i,j} \rho^{(d)}_N) - S(\rho^{(d)}_N)  = 2. \nonumber
\end{eqnarray}
In the case of $N$ odd, however, the state $\rho^{(d)}_N$ has $d^{N-1}$ eigenvalues equal to $\frac{1}{d^{N-1}}$, and by analogous calculations we get
\begin{eqnarray}
\mathcal{D}_N(\rho^{(d)}_N) = 1,
\end{eqnarray}
for $(N-1)$-uniform states.

Now we show that the $(N-1)$-uniform states are the only ones that can achieve $\mathcal{D}_N=2$. The requirement is
\begin{eqnarray}
\mathcal{D}_N &=& I(X_1:X_2|X_3...X_N) \nonumber \\
&=&  I(X_1:X_2X_3...X_N) - I(X_1:X_3...X_N) \nonumber \\
&=& 2, \label{APP_EQ_DN2}
\end{eqnarray}
where $X_i$ stands for individual subsystem. 
Since in the definition of $\mathcal{D}_N$ we minimize over all permutations, the same equation holds for all permutations of subsystems.
Due to subadditivity, the only way to satisfy (\ref{APP_EQ_DN2}) is
\begin{eqnarray}
I(X_1:X_3...X_N) &=& 0, \\
I(X_1:X_2X_3...X_N) &=& 2. \label{APP_EQ_MAX_I}
\end{eqnarray}
From the first equation we conclude that
\begin{equation}
\rho_{13...N} = \rho_1 \otimes \rho_{3...N},
\end{equation}
which also holds for all permutation of indices. 
After tracing out all but the 1st and 3rd subsystem, we arrive at
\begin{equation}
\rho_{13} = \rho_1 \otimes \rho_{3},
\end{equation}
which means that every pair of subsystems is described by a tensor product state. It follows that any $N-1$ particle subsystem is described by a simple tensor product, e.g.,
\begin{equation}\label{produ}
\rho_{13...N} = \rho_1 \otimes \rho_{3} \otimes \dots \otimes \rho_N.
\end{equation}
Using (\ref{APP_EQ_MAX_I}) we write
\begin{eqnarray}
S(X_1) - S(X_1|X_2X_3...X_N) = 2. 
\end{eqnarray}
Since for the quantum conditional entropy we have
\begin{eqnarray}
- S(X_1|X_2X_3...X_N) \leq S(X_1),
\end{eqnarray}
the bound is achieved if
\begin{eqnarray}
2 &=& S(X_1) - S(X_1|X_2X_3...X_N) \nonumber \\
&\leq& S(X_1)+S(X_1), \nonumber
\end{eqnarray}
i.e., for $S(X_1)=1$.
Hence, taking into account \eqref{produ}, all $N-1$ particle subsystems are maximally mixed, i.e., the total state is $(N-1)$-uniform.

\section{Dependence of Dicke states}

\label{formDicke}

We now present an analytical formula for $\mathcal{D}_N^e$ in $N$-qubit Dicke states with $e$ excitations.
For that state it is given by
\begin{eqnarray}
&\mathcal{D}_N(D_N^e) = {\binom{N}{e}}^{-1}  \Big[
-\frac{2 (N-1)! \log \left(\frac{e}{N}\right)}{(e-1)! (N-e)!}& \nonumber \\
&-2 \binom{N-1}{e} \log \left(1-\frac{e}{N}\right)+\binom{N-2}{e-2} \log \left(\frac{\binom{N-2}{e-2}}{\binom{N}{e}}\right)&\\
&+2 \binom{N-2}{e-1} \log \left(\frac{2 \binom{N-2}{e-1}}{\binom{N}{e}}\right)+\binom{N-2}{e} \log \left(\frac{\binom{N-2}{e}}{\binom{N}{e}}\right)\Big].& \nonumber
\end{eqnarray}
This comes from the fact that for a general Dicke state with $e$ excitations all one-partite reduced density matrices $\{\rho_i\}$ have the two non-zero eigenvalues $e/N$ and $(N-e)/N$, while all two-partite reduced states $\{\rho_{ij}\}$ have the three non-vanishing eigenvalues $e(e-1)/N(N-1)$, $2e(N-e)/N(N-1)$, and $(N-e-1)(N-e)/N(N-1)$. 
For $e$ as a function of the number of parties, $e = N/k$, in the limit of $N \to \infty$, the $N$-dependence converges to a finite value, i.e., $\mathcal{D}_N(D_N^e)$ tends to $2 (k-1)/k^2$.
The maximally achievable dependence of $1/2$ is reached for $e=N/2$.
For an arbitrarily chosen constant $e$ (e.g., for the W state, $e=1$), $\mathcal{D}_N(D_N^e)$ tends to 0 for $N \to \infty$.

These results allow to answer the following question:
If $\mathcal{D}_N \le 1$, are there local measurements on the subsystems with classical outcomes having conditional mutual information equal to $\mathcal{D}_N$?
The answer is negative.
We have optimized the conditional informations over local measurements for Dicke states with $N = 3,4$ and $0<e<N$,
and observed that the values obtained are always smaller than $\mathcal{D}_N$.

\section{Bounds on mutual $N$-dependence}
\label{ogr2}

\subsubsection{Bound on mixed states}

The subadditivity of quantum entropy states that for the reduced quantum states we have
\begin{eqnarray}
&S(\textrm{Tr}_j \rho) \leq S(\textrm{Tr}_{ij} \rho) + S(\rho_i),& \\
&S(\textrm{Tr}_i \rho) -  S(\rho_i) \leq S(\rho),&
\end{eqnarray}
where $\rho_i$ is the reduced state of the $i$-th particle.
Using the above inequalities we write
\begin{eqnarray}
\mathcal{D}_N(\rho) &\leq&  S(\textrm{Tr}_i \rho)  - S(\rho)   + S(\textrm{Tr}_j \rho)-S(\textrm{Tr}_{ij} \rho)  \nonumber \\
&\leq&  S(\rho_i)+S(\rho_i) \nonumber \\
&\leq&  2.
\end{eqnarray}

\subsubsection{Bounds on pure states}

Now we prove that for pure states we have $\mathcal{D}_N(\rho) \leq 1$. Note that due to Eq.~(\ref{pure}) from the main text we need to find the smallest mutual information $I(\rho_i:\rho_j)$, where $\rho_i$, $\rho_j$ are subsystems of the pure state $\rho$.
Consider 
\begin{eqnarray}
\lefteqn{I(\rho_i:\rho_j) + I(\rho_j:\rho_k)}\\
&& = S(\rho_i) + S(\rho_j) - S(\rho_{ij}) +  S(\rho_j) + S(\rho_k) - S(\rho_{jk}) \nonumber \\
&& \leq 2 S(\rho_j) \nonumber \\
&& \leq 2,
\end{eqnarray}
where the first inequality comes from the strong subadditivity of entropy
\begin{eqnarray}
S(\rho_i) + S(\rho_k) \leq S(\rho_{ij}) + S(\rho_{jk}).
\end{eqnarray}
Hence, this monogamy relation with respect to mutual information proves that there is always a bipartite subsystem with mutual information bounded by $1$.

\section{Experimental state generation and evaluation}
\label{APP_EXP}

The evaluated experimental states have been prepared using three different photonic setups which are detailed in Refs.~\cite{Kiesel2007,Krischek2010,Knips2016}.

The four-photon singlet state $\ket{\psi_4}$ was generated via a non-collinear type-II spontaneous parametric down conversion (SPDC) source. A pulsed UV laser with a central wavelength of $390~\text{nm}$ and an average power of about $600~\text{mW}$ from a frequency-doubled mode-locked Ti:sapphire laser was used to pump a $2~\text{mm}$-thick BBO ($\beta$-Barium Borate) crystal. For more details on this setup, see~\cite{Kiesel2007}. As a matter of fact, bosonic bunching also occurs for the emission of multi-photon states. Thereby entanglement beween four photons is obtained as in the state emitted by SPDC $\ket{\Phi^{\pm}}$-terms have a larger amplitude compared to $\ket{\Psi^{\pm}}$-terms.

For obtaining the Dicke and the no-correlation states, the SPDC crystal ($1~\text{mm}$-thick BBO, type II) was placed inside a femto-second UV-enhancement resonator. After the colinear creation of an equal number of photons by type-II SPDC, they have been distributed to 4 or 6 analyser stations. Conditioning on detecting a photon in every station, the symmetric Dicke states $D_4^2$ and $D_6^3$ can be observed. There, we obtained an average UV power of up to $8.2~\text{W}$ at a repetition rate of $81~\text{MHz}$ (see ~\cite{Krischek2010} for more details on the setup).  The three-photon state $\rho_{\textrm{nc},3}$ was obtained by tracing out one particle from $D_4^{2}$, whereas projection of one photon's polarization onto the horizontal direction and tracing this photon out provides us the Dicke state $D_3^{1}$. The tomographic data of $D_4^{2}$ was originally taken for Ref.~\cite{Schwemmer2015}. 
The five photon state $\rho_{\textrm{nc},5}$ was deduced from $D_6^{3}$.

The four-qubit $\mathrm{GHZ}_4$ state as well as the linear cluster state $L_4$ were prepared in a two-photon setup, where the four qubits were encoded in two degrees of freedom per photon, the polarization and the path. A type-I SPDC source (two crossed type-I BBO crystals pumped with a $402~\text{nm}$ continuous-wave laser at $60~\text{mW}$) generated polarization-entangled photon pairs. Using polarizing beam splitters the two polarization qubits were coupled to the path degree of freedom inside two displaced Sagnac interferometers. Right after a polarising beamsplitter, the four qubit $\mathrm{GHZ}_4$ state is created, and the combined manipulation of polarisation and path qubits in the Sganac interferometer  enables the transformation of this state to other multi-qubit entangled states~\cite{Knips2016}.

Genarally, the evaluation of the dependence is based on tomographic data. As the dependence is calculated from the eigenvalues of the reconstructed state and its marginals, a careful treatment of the eigenvalues of the reconstructed states is crucial. However, as the linearly reconstructed states feature negative eigenvalues~\cite{Schwemmer_2015}, we instead resort to a simple tomographic reconstruction based on the findings in~\cite{Krischek2010}. We model the state by $\hat\varrho=\sum_i \lambda_i |\psi_i \rangle \langle \psi_i| + \lambda_{\perp} \openone_{\perp}$, where $\lambda_i$ ($\ket{\psi_i}$) are the unaltered eigenvalues (eigenstates) from the direct state estimate which are well above a noise threshold, $\lambda_{\perp}$ is the sum of the noise eigenvalues where $\openone_{\perp}$ denotes the identity matrix in the space spanned by the noise eigenstates.

\end{document}